\documentclass[11pt]{article}

\usepackage[a4paper, bottom =3.5cm, top=3cm]{geometry}
\usepackage{graphicx} 
\usepackage{cite}
\usepackage{tikz}
\usepackage{setspace}
\usepackage{amsmath, amsthm, amssymb}
\usepackage{authblk}

\newtheorem{thm}{Theorem}[section]
\newtheorem{cor}[thm]{Corollary}
\newtheorem{lem}[thm]{Lemma}
\newtheorem{ex}[thm]{Example}
\newtheorem{defn}[thm]{Definition}

\newcommand{\ie}{\textit{i.e.}}
\newcommand{\eg}{\textit{e.g.}}

\def\Hom{\mathop\mathrm{Hom}\nolimits}
\def\Aut{\mathop\mathrm{Aut}\nolimits}
\def\inj{\mathop\mathrm{inj}\nolimits}
\def\hom{\mathop\mathrm{hom}\nolimits}
\def\injstar{\inj^*}
\def\homstar{\hom^*}
\def\vH{v_H}
\def\GF{\mathrm{GF}}

\def\Hcol#1{$\#_{#1}H$-\textsc{Colouring}}
\def\parityHcol{$\oplus H$-\textsc{Colouring}}
\def\parityTcol{$\oplus T$-\textsc{Colouring}}
\def\parityHoneCol{$\oplus H_1$-\textsc{Colouring}}
\def\modSAT#1{$\#_k$\textsc{Sat}}
\def\parityIndSet{$\oplus$\textsc{IndSet}}
\def\parityP{\oplus\mathrm{P}}
\def\numP{\mathrm{\#P}}

\def\modnumP#1{\#_{#1}\mathrm{P}}
\def\SAT{\textsc{Sat}}
\def\paritySAT{\oplus\textsc{Sat}}
\def\redT{\leq^\mathrm{T}_\mathrm{p}}
\let\phi=\varphi

\makeatletter
\def\prob#1#2#3{\goodbreak\begin{list}{}{\labelwidth\z@ \itemindent-\leftmargin
                        \itemsep\z@  \topsep6\p@\@plus6\p@
                        \let\makelabel\descriptionlabel}
                \item[\it Name.]#1
                \item[\it Instance.]#2
                \item[\it Output.]#3
                \end{list}}
\makeatother

\title{The complexity of parity graph homomorphism:\\ an initial investigation}
\author{John Faben%
\thanks{Supported by EPSRC grant EP/E064906/1 
``The Complexity of Counting in Constraint Satisfaction Problems''} }
\author{Mark Jerrum%
\thanks{Mailing address: Mile End Road, London E1 4NS, UK.
Supported by EPSRC grant EP/I011935/1 ``Computational Counting''}}
\affil{School of Mathematical Sciences\\Queen Mary, University of London}

\begin{document}

\maketitle

\begin{abstract}
Given a graph $G$, we investigate the question of determining the parity of the number of homomorphisms from $G$ to some other fixed graph $H$. We conjecture that this problem exhibits a complexity dichotomy, such that all parity graph homomorphism problems are either polynomial-time solvable or $\parityP$--complete, and provide a conjectured characterisation of the easy cases. 

We show that the conjecture is true for the restricted case in which the graph~$H$ is a tree, and provide some tools that may be useful in further investigation into the parity graph homomorphism problem, and the problem of counting homomorphisms for other moduli.
\end{abstract}

\section{Graph homomorphism} \label{secGraphHom}

Graph homomorphism is a natural generalisation of graph colouring, in which the restrictions on adjacencies between colours can be more general than in the usual graph colouring problem. A homomorphism from a graph $G$ to a graph $H$ is an edge-preserving map between the vertices (see Definition \ref{defHomomorphism}). It is sometimes referred to as an $H$-colouring (where the target graph for the homomorphism is $H$). Ordinary graph colouring is the special case of homomorphisms into the complete graph.

\begin{defn} \label{defHomomorphism}
 A homomorphism from a graph $G$ into another graph $H$ is a map $\phi: V(G) \rightarrow V(H)$ satisfying the property that if\/ $(u,v) \in E(G)$ then $(\phi(u),\phi(v)) \in E(H)$.  The set of homomorphisms from $G$ to~$H$ is denoted
$\Hom(G,H)$.
\end{defn}

\begin{ex} 
A homomorphism from a graph $G$ to the complete graph $K_n$ is a (proper, vertex)
$n$-colouring of $G$.
\end{ex}

\begin{ex}
Let $H_1$ be the graph with vertex set $\{a,b\}$, an edge joining $a$ and $b$, and a loop at~$b$. 
A homomorphism from a graph $G$ to $H_1$ can be considered as an independent set of $G$. The vertices mapped to vertex~$a$ form an independent set (as none of them can be pairwise adjacent) and, conversely, given an independent set, it is possible to map the vertices of the independent set to~$a$ and the vertices of its complement to~$b$. So there is a natural one-to-one correspondence between homomorphisms to $H_1$ and independent sets.
\end{ex}

For the purposes of this paper, both $G$ and $H$ are allowed to have loops on their vertices, but not multiple edges. To reduce the potential for confusion, we will usually refer to the vertices of $H$ as ``colours'', reserving the word ``vertex'' for vertices of $G$. 

Fix a target graph~$H$.  There are a number of computational problems
of the form:  given an instance (graph) $G$ return some information about $\Hom(G,H)$.
The most basic one is the decision problem, which asks if $\Hom(G,H)$ non-empty.
Each $H$ specifies a particular decision problem;  for example, if $H$ is the triangle, 
the problem is to decide if $G$ is 3-colourable.  The goal is then to classify the complexity 
of the computational problem in terms of the graph~$H$.  The ideal is to identify a dichotomy,
i.e., a partition of graphs $H$ into those that specify tractable problems and those that
specify intractable ones.

The complexity of the decision version of the graph homomorphism problem 
was completely classified by Hell and Ne\v{s}et\v{r}il in \cite{HellNesetril90}. For a given graph $H$, deciding whether an arbitrary graph has a homomorphism to $H$ can be done in polynomial time if $H$ has a loop or is bipartite. Hell and Ne\v{s}et\v{r}il showed that this decision problem is NP-complete in all other cases. 

It is also natural to consider the counting problem, which asks for the cardinality of $\Hom(G,H)$, which we denote by $\hom(G,H)$.
The problem of exactly counting the number of homomorphisms to a fixed graph~$H$ was considered by Dyer and Greenhill~\cite{DyerGreenhill2000}, who gave a complete characterisation, again with a dichotomy theorem: the counting problem 
is polynomial-time solvable if $H$ is either a complete graph with loops everywhere or a complete bipartite graph without loops, and it is $\numP$-complete otherwise. 

The result of Dyer and Greenhill has been extended in many different directions
by various authors.  One possibility is to specify weights $w:E(H)\to\mathbb{C}$
for the edges of~$H$;  this edge-weighting naturally induces a weighing of homomorphisms~$\phi$
from $G$ to~$H$, by taking a product of weights $w(\phi(u),\phi(v))$ over edges $\{u,v\}$ of~$G$.
In the weighted setting, one can express partition functions of models in statistical physics.
Note that the unweighted form of the problem can be recovered by restricting weights to be $\{0,1\}$.
Bulatov and Grohe~\cite{BulatovGrohe05} exhibited a dichotomy for non-negative real weights,
which was extended to arbitrary real weights by 
Goldberg, Grohe, Jerrum and Thurley~\cite{GGGT10}, and then on to complex weights 
by Cai, Chen and Lu~\cite{CaiChenLu13}.  The 
massive further generalisation to Constraint Satisfaction Problems (CSPs) was undertaken 
by several authors (\eg, Bulatov~\cite{Bulatov08} and Dyer and Richerby~\cite{DyerR11}), 
culminating in the complex weighted case by Cai and Chen~\cite{CaiChen12}.
See Chen's survey for more details~\cite{ChenSurvey11}.

In this paper, we shall mostly be concerned with the problem of determining the cardinality of $\Hom(G,H)$
modulo~$k$, for a positive integer~$k$, with a special emphasis on $k=2$, i.e., determining wither 
the number of $H$-colourings is odd or even.  For $k\geq 2$ and $n$ an integer, denote by $[n]_k$
the residue class of~$n$ \textit{modulo}~$k$.  We can of course identify these classes with the 
integers $\{0,1,\ldots,k-1\}$.  Formally, our computational problem is the following. 
\prob%
{\Hcol{k}.}%
{An undirected graph $G$.}%
{$[\hom(G,H)]_k$, i.e., the number of $H$-colourings of $G$ \textit{modulo} $k$.}
Since the case $k=2$ is of special significance, we introduce \parityHcol{} 
as a synonym for \Hcol2.

We give a dichotomy theorem for \parityHcol{} in the case where $H$ is a tree: either 
\parityHcol{} is $\parityP$-complete or it can be solved in polynomial time.
(See Theorem~\ref{thmTreeDichotomy}.) 
Informally, $\parityP$ is the class of problems that can be expressed in terms
of deciding the parity of the number of accepting computations of a non-deterministic Turing
machine;  see Section~\ref{sec:complex} for a precise definition. 
The proof of the dichotomy is based on a reduction system which transforms $H$ to a ``reduced form''
of equivalent complexity.
Since it is easy to decide the complexity of \parityHcol{} for reduced forms, we obtain not
only the dichotomy result, but also an effective procedure for deciding the dichotomy.
We conjecture that the same reduction system decribes a complexity dichotomy for general graphs. 
Although this conjecture remains open in general, G\"obel, Goldberg and Richerby~\cite{GobelEtAl13} have 
extended our result by showing that the conjecture holds for cactus graphs.

Finally we draw attention to some existing work in the general area of modular counting.
The complexity of modular counting problems has been studied for at least three decades,
early contributions being made by
Valiant~\cite{Valiant79} and Papadimitriou and Zachos\cite{Papadimtriou82}.  
One of the more striking results, is that of Valiant~\cite{Valiant06},
who provides an example of a counting problem that is unexpectedly easy \textit{modulo}~7,
though hard \textit{modulo}~2.
It is worth noting that modular CSPs have been studied, \eg, 
by Faben~\cite{Faben2008} and Guo, Huang, Lu and Xia~\cite{GuoHLX11}. 
This work is both more
general, in the sense of being set within the wider context of CSPs, but also more restrictive,
in that it relates to the two-element (Boolean) domain only.

\section{Modular counting complexity}\label{sec:complex}

\subsection{The classes $\modnumP{k}$} \label{secClasses}
In this section, we formally define the counting classes that we will use in this paper.

A classical counting problem can be considered as a function taking a problem
instance to the number of solutions associated with that instance.  
When counting is done \textit{modulo\/} some number $k\geq2$,
it is possible to view the problem from two somewhat different 
standpoints.  On the one hand there is the decision or language view, where 
the task is to determine whether the number of solutions is different from 0, 
\textit{modulo} $k$.  On the other is the function view, where the 
task is to compute the residue, modulo $k$, of the number of solutions.  
Both views have been taken in earlier work, and the distinctions between 
them have been examined by Faben~\cite{Faben2012}.  

In the current context, the function view seems more natural.  
We work within a class $\modnumP{k}$ of computational 
problems which is the modular analogue of the
classical class $\numP$ of counting problems.  Informally, $\modnumP{k}$ contains
functions that can be expressed as the residue, \textit{modulo} $k$, of the
number of accepting computational of a nondeterministic polynomial-time Turing
machine.

Let $\Sigma$ be a finite alphabet over which we agree to encode problem instances.
\begin{defn}
Let $M$ be a non-deterministic Turing Machine. 
Denote by $\#\mathrm{acc}_M(x)$ the number of accepting paths of the machine~$M$ 
on the input $x\in\Sigma^*$.
The class $\numP$ consists of all functions $f:\Sigma^*\to\mathbb{N}$ that can be 
expressed as $f(x)=\#\mathrm{acc}_{M}(x)$ for some non-deterministic polynomial-time 
Turing Machine~$M$.
The class $\modnumP{k}$ consists of all functions $f:\Sigma^*\to\{0,1,\ldots,k-1\}$ 
that can be expressed as $f(x)=[\#\mathrm{acc}_M(x)]_k$.
\end{defn}
In this paper, we are concerned particularly with the case $k=2$, and we follow 
other authors in using $\parityP$ as a synonym for $\modnumP2$~\cite{Papadimtriou82}.

Given a counting problem in $\numP$, say $\#A$, we write \#$_kA$ for the \#$_k$P problem of determining the number of solutions to $A$ \textit{modulo} $k$. So while \#$A:\Sigma^* \rightarrow \mathbb{N}$ is a function defined from strings to the natural numbers, \#$_kA:\Sigma^* \rightarrow \{0, \ldots, k-1\}$ is the function from strings to the integers \textit{modulo} $k$ defined by \#$_kA(x) \equiv \textrm{\#}A(x) \pmod k$.  
As an example, \modSAT{k} is the problem of determining
the number of satsifying assignments to a CNF Boolean formula, \textit{modulo}~$k$.
Naturally, $\paritySAT$ is the special case $k=2$ of this problem.

\subsection{Completeness}

Again, in an analogy with $\numP$-completeness, we define the notion of $\modnumP{k}$-completeness with respect to polynomial-time Turing reducibility. Essentially, a problem $A$ is $\modnumP{k}$-hard if every problem in $\modnumP{k}$ can be solved in polynomial time given an oracle for~$A$.

\begin{defn} \label{defnTuringReduction}
We say that a problem $B$ is {\em polynomial-time Turing reducible} to a problem $A$ if problem~$B$ can be solved in polynomial time using an oracle for problem~$A$. We write $B \redT A$.
\end{defn}
 
\begin{defn} \label{defnHardness}
A counting problem $A$ is {\em$\modnumP{k}$-hard} if, for every problem $B$ in $\modnumP{k}$, $B \redT A$, \ie, if every problem in $\modnumP{k}$ is polynomial-time Turing reducible to $A$. 
It is {\em$\modnumP{k}$-complete} if, in addition, $A$ is in $\modnumP{k}$.
\end{defn}

As one might expect, the modular counting versions of \SAT, namely \modSAT{k} for $k\geq2$,
are examples of $\modnumP{k}$-complete problems for all $k$. This can be easily seen as the usual reduction in Cook's Theorem, showing that SAT is NP-complete is parsimonious (\ie, preserves the number of solutions), and so certainly preserves the number of solutions \textit{modulo} $k$ for all $k$.

As mentioned above, the complexity of exactly counting the number of homomorphisms to a given graph $H$ was characterised by Dyer and Greenhill. They proved the following theorem.

\begin{thm} [Dyer and Greenhill \cite{DyerGreenhill2000}] \label{thmDyerGreenhill} 
If a graph $H$ is a complete bipartite graph with no loops or a complete graph with loops everywhere, then exactly counting $H$-colourings can be done in polynomial time. Otherwise, the problem is $\numP$-complete.
\end{thm}

Clearly, if the number of homomorphisms to a graph $H$ can be counted exactly in polynomial time, then the parity can be determined in polynomial time. We will show that there are some cases in which symmetries of~$H$ can make the related modular counting problem easy, even when the exact counting problem is $\numP$-hard.

\section{Reduction by involutions} \label{secReductionByInvolutions}

We will also need to discuss the automorphism groups of graphs. There will be particular reference to automorphisms of order 2, or involutions.

\begin{defn}
An {\em automorphism} of a graph $G$ is an injective homomorphism from $G$ to itself. 
In other words, an automorphism of a graph $G$ is a permutation~$\sigma$ of the vertices of $G$ such that $\{\sigma(u),\sigma(v)\} \in E(G) \iff \{u,v\} \in E(G)$. 
If $\sigma$ has order 2, i.e., $\sigma$ is not the identity but $\sigma\circ\sigma$ is, then we say that $\sigma$ is an {\em involution} of~$G$.
\end{defn}

As hinted at earlier, our approach is based on a reduction system for graphs $H$ that preserves the complexity of 
the problem \parityHcol.  The reductions are defined in terms of the automorphisms of~$H$.

\begin{defn} \label{defSigmaH}
Let $H$ be a graph, and $\sigma$ an automorphism of $H$. We denote by $H^\sigma$ the subgraph of $H$ induced by the fixed points of $\sigma$.
\end{defn}

\begin{lem} \label{lemInvolutions}
If $H$ is a graph, and $\sigma$ an involution of $H$, the number of $H$-colourings of any graph~$G$ is congruent \textit{modulo} 2 to the number of $H^\sigma$-colourings of $G$.
\end{lem}

\begin{proof}
We will in fact show that the number of $H$-colourings of $G$ which are not $H^\sigma$-colourings is even,
which is equivalent to saying that see the number of $H$-colourings which use at least 
one colour in $V(H) \backslash V(H^\sigma)$ is even.  

To see this,
we partition the set of such colourings into subsets of size two. 
The basic idea here is that to each colouring which uses at least colour in $V(H)\backslash V(H^\sigma)$
we can associate the colouring gained by first applying $\sigma$ to~$H$ and then colouring $G$. 
Formally, given any colouring $\phi: V(G) \rightarrow V(H)$, consider the alternative colouring $\sigma \circ \phi$. This is still an $H$-colouring of $G$, as both $\sigma$ and $\phi$ are edge-preserving. It is different from $\phi$ as there is some vertex $v \in G$ such that $\phi(v) \in V(H) \backslash V(H^\sigma)$, and so $\sigma(\phi(v)) \neq \phi(v)$. On the other hand $\sigma \circ \sigma \circ \phi$ is just $\phi$, as $\sigma$ is an involution. So $\sigma$ acts as an involution on the set of $H$-colourings of $G$ which use at least one colour from $V(H) \backslash V(H^\sigma)$. Since this involution has no fixed points, 
the size of this set must be even. 
\end{proof}

Note that the above argument does not rely on any special properties of the modulus~2 beyond 
the fact that it is prime.  

\begin{thm}\label{thm:inv}
For any prime $p$, if $H$ is a graph, and $\sigma$ an automorphism of $H$ of order~$p$, the number of $H$-colourings of any graph~$G$ is congruent \textit{modulo}~$p$ to the number of $H^\sigma$-colourings of $G$.
\end{thm}

It is not just the proof that fails for a composite modulus~$k$.  The complete graph $K_5$ on five vertices
has an automorphism of order~6 that moves all the vertices, 
but it is not true that for every graph~$G$ that the number 
of 5-colourings of~$G$ is divisible by~6. 

We define the following reduction system on the set of unlabelled graphs.

\begin{defn} \label{defReduction}
The binary relation $\rightarrow_k$ on graphs is defined as follows.  For graphs $H$ and $K$, the relation $H \rightarrow_k K$ 
holds iff there exists an automorphism~$\sigma$ of~$H$, of order~$k$, such that $H^\sigma = K$. 
If there exists a sequence of graphs $H_1, H_2, \ldots, H_\ell$ such that $H \rightarrow_k H_1 \rightarrow_k \ H_2 \rightarrow_k \ldots \rightarrow_k H_\ell  = K$,
we write $H \rightarrow_k^* K$ and say that {\em $H$ reduces to $K$ by automorphisms of order~$k$}. (If $k=2$, we say that 
{\em $H$ reduces to~$K$ by involutions}.)  If $K$ has no automorphisms of order~$k$ we say that $K$ is a {\em reduced form} associated with the graph~$H$.
\end{defn}

\begin{ex}
In Figure \ref{Something} we give an example of a graph $H$, along with two ways of reducing $H$ by involutions. On the right-hand side we reduce $H$ by using the involution~$\sigma$ which swaps each of the pairs of vertices $a$ and $e$, $b$ and $f$, $c$ and $d$, leaving behind only the involution-free graph on the vertices $g$ and~$h$. On the left-hand side, we begin with the involution $\tau$ which swaps $e$ and $f$, and have to reduce the resulting graph by involutions twice more before we get to the involution-free graph $((H^\tau)^\upsilon)^\eta$ which is isomorphic to the graph $H^\sigma$.  This is not a coincidence.
We will see in Theorem \ref{thmConfluence} that reduced forms are unique.
\end{ex}

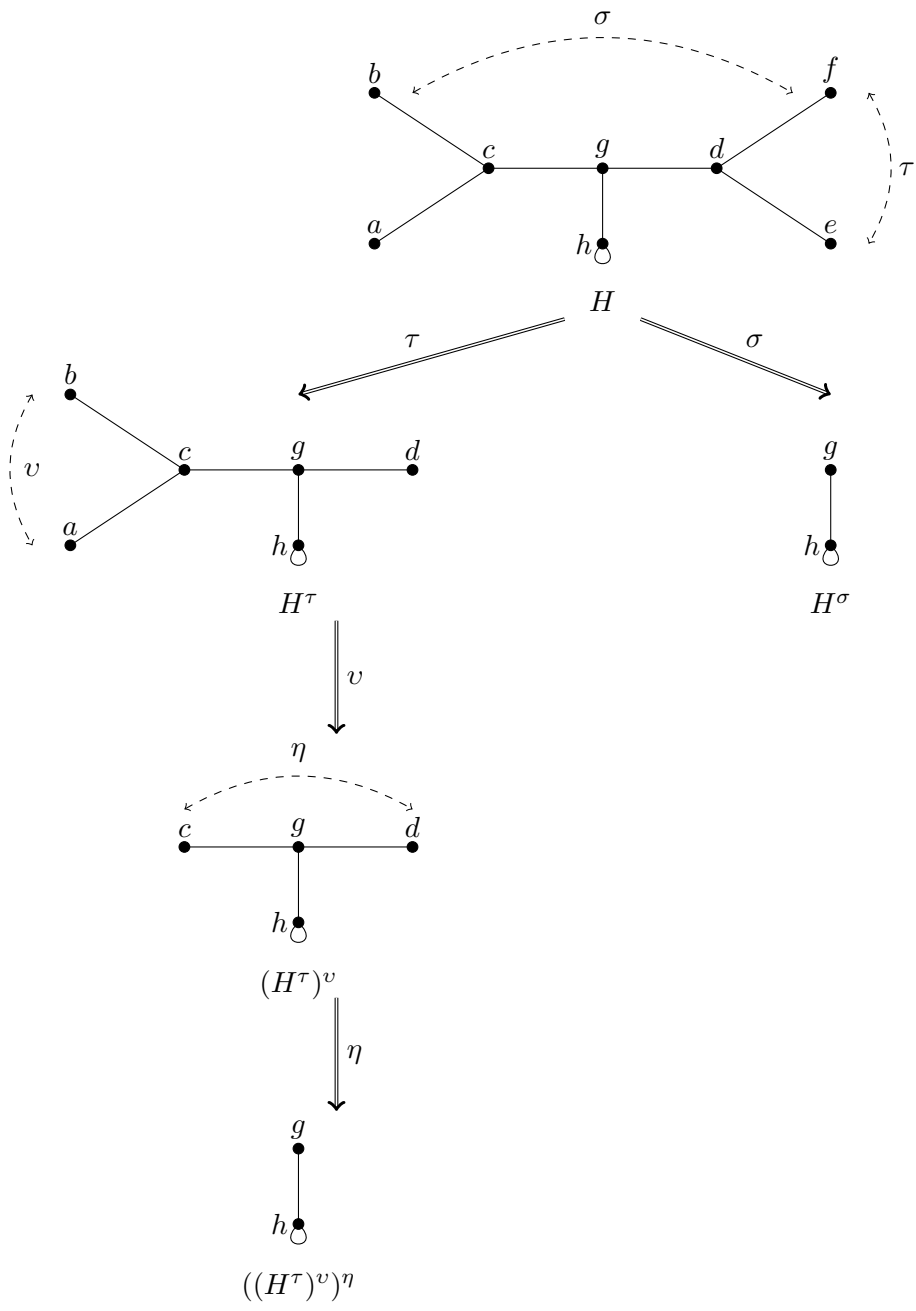
\begin{figure} 
\begin{center}
\begin{tikzpicture}

\path (1,1) coordinate (a);
\path (1,1) node[above] (alab) {$a$};

\path (1,3) coordinate (b);
\path (1,3) node[above] (blab) {$b$};

\path (2.5,2) coordinate (c);
\path (2.5,2) node[above] (clab) {$c$};

\path (4,2) coordinate (g);
\path (4,2) node[above] (glab) {$g$};

\path (5.5,2) coordinate (d);
\path (5.5,2) node[above] (dlab) {$d$};

\path (7,1) coordinate (e);
\path (7,1) node[above] (elab) {$e$};

\path (7,3) coordinate (f);
\path (7,3) node[above] (flab) {$f$};

\path (4,1) coordinate (h);
\path (4,1) node[left] (hlab) {$h$};
\draw[out=-45, in=-135, loop] (h) to ();

\path (4,4) node (sigma) {$\sigma$};
\path (8,2) node (tau) {$\tau$};

\draw [<->,dashed] (1.5,3) to[bend left] (6.5,3);
\draw [<->,dashed] (7.5,3) to[bend left] (7.5,1);

\draw [fill] (a) circle (2pt);
\draw [fill] (b) circle (2pt);
\draw [fill] (c) circle (2pt);
\draw [fill] (d) circle (2pt);
\draw [fill] (e) circle (2pt);
\draw [fill] (f) circle (2pt);
\draw [fill] (g) circle (2pt);
\draw [fill] (h) circle (2pt);

\draw (g)--(h);
\draw (a)--(c)--(b);
\draw (c)--(g)--(d)--(e);
\draw (g)--(g);
\draw (d)--(f);

\path (4,0.5) node[below] (G) {$H$};


\draw [->,double] (3.5,0) -- (0,-1);
\path (1.5,-0.5) node[above] (tauarrow) {$\tau$};

\draw [->,double] (0.5,-4) -- (0.5,-5.5);
\path (0.5,-4.75) node[right] (upsilonarrow) {$\upsilon$};

\draw [->,double] (0.5,-9) -- (0.5,-10.5);
\path (0.5,-9.75) node[right] (etanarrow) {$\eta$};

\draw [->,double] (4.5,0) -- (7,-1);
\path (6,-0.5) node[above] (sigmaarrow) {$\sigma$};

\path (-3,-3) coordinate (adash);
\path (-3,-3) node[above] (alabdash) {$a$};

\path (-3,-1) coordinate (bdash);
\path (-3,-1) node[above] (blabdash) {$b$};

\path (-1.5,-2) coordinate (cdash);
\path (-1.5,-2) node[above] (clabdash) {$c$};

\path (0,-2) coordinate (gdash);
\path (0,-2) node[above] (glabdash) {$g$};

\path (1.5,-2) coordinate (ddash);
\path (1.5,-2) node[above] (dlabdash) {$d$};

\path (0,-3) coordinate (hdash);
\path (0,-3) node[left] (hlabdash) {$h$};
\draw[out=-45, in=-135, loop] (hdash) to ();

\draw [<->,dashed] (-3.5,-1) to[bend right] (-3.5,-3);
\path (-3.5,-2) node (upsilon) {$\upsilon$};

\draw [fill] (adash) circle (2pt);
\draw [fill] (bdash) circle (2pt);
\draw [fill] (cdash) circle (2pt);
\draw [fill] (ddash) circle (2pt);
\draw [fill] (gdash) circle (2pt);
\draw [fill] (hdash) circle (2pt);

\draw (gdash)--(hdash);
\draw (adash)--(cdash)--(bdash);
\draw (cdash)--(gdash)--(ddash);

\path (0,-3.5) node[below] (Gtau) {$H^\tau$};

\path (7,-2) coordinate (gsigma);
\path (7,-2) node[above] (glabsigma) {$g$};

\path (7,-3) coordinate (hsigma);
\path (7,-3) node[left] (hlabsigma) {$h$};
\draw[out=-45, in=-135, loop] (hsigma) to ();

\draw [fill] (gsigma) circle (2pt);
\draw [fill] (hsigma) circle (2pt);

\draw (gsigma)--(hsigma);

\path (7,-3.5) node[below] (Gsigma) {$H^\sigma$};


\path (-1.5,-7) coordinate (cdash1);
\path (-1.5,-7) node[above] (clabdash1) {$c$};

\path (0,-7) coordinate (gdash1);
\path (0,-7) node[above] (glabdash1) {$g$};

\path (1.5,-7) coordinate (ddash1);
\path (1.5,-7) node[above] (dlabdash1) {$d$};

\path (0,-8) coordinate (hdash1);
\path (0,-8) node[left] (hlabdash1) {$h$};
\draw[out=-45, in=-135, loop] (hdash1) to ();

\draw [<->,dashed] (-1.5,-6.5) to[bend left] (1.5,-6.5);
\path (0,-6) node[above] (vau) {$\eta$};

\draw [fill] (cdash1) circle (2pt);
\draw [fill] (ddash1) circle (2pt);
\draw [fill] (gdash1) circle (2pt);
\draw [fill] (hdash1) circle (2pt);

\draw (gdash1)--(hdash1);
\draw (cdash1)--(gdash1)--(ddash1);

\path (0,-8.5) node[below] (Gtau) {$(H^\tau)^\upsilon$};

\path (0,-11) coordinate (gdash2);
\path (0,-11) node[above] (glabdash2) {$g$};

\path (0,-12) coordinate (hdash2);
\path (0,-12) node[left] (hlabdash2) {$h$};
\draw[out=-45, in=-135, loop] (hdash2) to ();

\draw [fill] (gdash2) circle (2pt);
\draw [fill] (hdash2) circle (2pt);

\draw (gdash2)--(hdash2);

\path (0,-12.5) node[below] (Gtauagain) {$((H^\tau)^\upsilon)^\eta$};
\end{tikzpicture}
\end{center}

\caption{An example of a graph $H$ with the sequence of reductions we get from $H$ if we start with each of the involutions $\sigma$ and $\tau$.}
\label{Something}
\end{figure}

To make further progress, we need to assume $k=p$ is prime.  Eventually, we will further restrict 
attention to the case $p=2$.  However, we state and prove some intermediate results for
a general prime~$p$, as they may be of use in further explorations of modular counting
problems.  

Theorem~\ref{thm:inv} says that in classifying the complexity of \Hcol{p} problems, it is enough to
restrict attention to graphs~$H$ that are reduced forms, \ie, that do not have any automorphisms of 
order~$p$.  This is enough for the proof of the main dichotomy result, 
but it is an interesting fact that reduced forms are unique.  In any case, the concepts used 
in the proof of uniqueness of the reduced form will be needed later.  

\begin{thm} \label{thmConfluence}
Given a graph $G$, and a prime $p$ there is (up to isomorphism) exactly one graph $G^*$ such that $G^*$ has no automorphisms of order~$p$ and $G \rightarrow_p^* G^*$.
\end{thm}

The proof, which uses the concept of ``Lov\'asz vector'' of a graph, is presented in the next section.
We can now state main result.

\begin{thm} \label{thmTreeDichotomy}
If $H$ is a tree, then \parityHcol{} is $\parityP$-complete if the reduced form obtained by reducing $H$ by involutions is non-trivial,
\ie, has more than one vertex. Otherwise it is solvable in polynomial time.
\end{thm}

We conjecture that this result holds for graphs in general.  The conjecture is unresolved, though 
G\"obel, Goldberg and Richerby~\cite{GobelEtAl13} recently extended our result from trees 
to cactus graphs.  One could extend the conjecture to \Hcol{p}, for primes $p>2$.  Specifically,
one might conjecture that, for each $p$, the set of reduced forms~$H$ corresponding to polynomial-time 
cases of \Hcol{p} is finite (and that all other reduced forms correspond to $\numP$-complete cases).
However, we do not go that far here.

\subsection{The Lov\'asz vector of a graph}
We need a modular version of the Lov\'asz vector~\cite[\S2.3]{HellNesetril90} of a graph.

\begin{defn} \label{defnLovaszVector}
Let $p$ be a prime, and $G_1, G_2, \ldots$ be a fixed enumeration of all pairwise non-isomorphic graphs. 
(Thus every graph is isomorphic to exactly one graph in the sequence.)
The {\em mod-$p$ Lov\'asz vector} of a graph $H$ is the sequence 
$\big([\hom(G_i,H)]_p:i \geq 1\big)$.
\end{defn}

We show that the mod-$p$ Lov\'asz vector determines a graph, provided the graph has no
automorphisms of order~$p$.  First recall some elementary facts about groups. 

\begin{thm} [Cauchy's Group Theorem]
If a prime $p$ divides the order of a finite group $G$, then $G$ contains at least one element of order $p$.
\end{thm}

\begin{thm} [Lagrange's Theorem]
 For any finite group $G$ the order of any subgroup of $G$ divides the order of $G$.
\end{thm}

It follows that:

\begin{lem} \label{lemInvolutionIffEvenOrder}
For any prime $p$, a graph has an automorphism of order~$p$ if and only if the order of its automorphism group is divisible by~$p$.
\end{lem}

\begin{proof}
The automorphisms of a graph form a group. If this group contains an element of order~$p$, 
then the order of the automorphism group is divisible by~$p$, by Lagrange's Theorem. 
If the order of the automorphism group is divisible by $p$, then it contains an automorphism of order $p$ by Cauchy's Group Theorem.
\end{proof}

Now the claim.

\begin{lem} \label{lemModLovasz}
Suppose $p$ is a prime, and  $H$ and $H'$ are two graphs,
neither of which has an automorphism of order~$p$.
Then $H$ and $H'$ are isomorphic if and only if they have the same mod-$p$ Lov\'asz vector.
\end{lem}

\begin{proof}
Clearly the condition is necessary: two isomorphic graphs have the same mod-$p$ Lov\'asz vector. 
Now we need to prove that it is sufficient.  
This proof is similar to the proof of Theorem 2.11 in Hell and Ne\v set\v ril's monograph \cite{HellNesetril04}. 

So suppose $H$ and $H'$ have the same mod-$p$ Lov\'asz vector, that is
\begin{equation} \label{eqnHoms}
 \hom(G,H) \equiv \hom(G,H')\pmod p
\end{equation}
for all graphs $G$.
We first observe that, in order to show that $H$ and $H'$ are isomorphic, 
it is sufficient to prove that for every graph~$G$:
\begin{equation} \label{eqnInjections}
 \inj(G,H) \equiv \inj(G,H')\pmod p,
\end{equation}
where $\inj(G,H)$ denotes the number of injective homomorphisms from $G$ to~$H$.
To see this, first take $G = H$ in the above congruence~(\ref{eqnInjections}). 
The left hand side of the congruence is just the order of the automorphism group of~$H$,
which, by Lemma~\ref{lemInvolutionIffEvenOrder}, is not congruent to~0 {\em modulo}~$p$.
Therefore, the right hand side, $\inj(H,H')$, is also different from from~0 {\em modulo}~$p$
and, in particular, there exists an injective homomorphism from $H$ to~$H'$. 
Similarly, if we take $G=H'$ we find an injective homomorphism the other way, 
and thus an isomorphism between $H$ and~$H'$.

We will prove that equation (\ref{eqnHoms}) implies equation (\ref{eqnInjections}) 
by induction on the size of~$G$. If $G$ only has one vertex then every homomorphism from~$G$ 
to any other graph is injective, so the equality holds. Now assume that the equality 
is true for all graphs that have fewer vertices than~$G$. 
The proof strategy is essentially to count those homomorphisms which are not injections, 
and show that there are the same number of these, so there must be the same number of injective homomorphisms.

For a partition $\Theta=\{S_i:i\in I\}$ of the vertex set $V(G)$ of a graph $G$, 
define the quotient graph $G/\Theta$ as follows.  The vertex set of $G/\Theta$ is the 
index set~$I$.  There is an edge between $i,j\in I$ in $G/\Theta$ iff there is some edge
joining a vertex in $S_i$ to a vertex in~$S_j$ in~$\Theta$.  (It may happen that 
$i=j$, in which case $G/\Theta$ has a loop at~$i$.)
The number of homomorphisms is equal to the number of injective homomorphisms 
plus the number of homomorphisms which are not injective. 
In order to count the number of homomorphisms which are not injective, 
we consider the different ways in which the colours of $H$ can be assigned to vertices of $G$.

A colouring of $G$ with $H$ induces a partition of $G$ in the obvious way, 
with vertices which are given the same colour assigned to the same part of the partition. 
If we call this partition $\Theta$, then any $H$-colouring of $G$ can be considered as an injective $H$-colouring 
of $G / \Theta$, since each vertex of $G / \Theta$ is associated with exactly one colour from~$H$. 
Let $\iota$ be the partition consisting of a single block for each vertex 
(\ie, the partition associated with injective homomorphisms from $G$ to~$H$). Then we have both
\begin{align*} 
\hom(G,H) &=\inj(G,H) + \sum_{\Theta \neq \iota} \inj(G / \Theta, H) \\
\noalign{\noindent and}
\hom(G,H')&= \inj(G,H') + \sum_{\Theta \neq \iota} \inj(G / \Theta, H').
\end{align*}
Since $G / \Theta$ is necessarily smaller than $G$ if $\Theta \neq \iota$, 
we know by the induction hypothesis that $\inj(G / \Theta, H) \equiv \inj(G / \Theta, H')\pmod p$,
and since $\hom(G,H) \equiv \hom(G,H')\pmod p$ by assumption, we do have $\inj(G,H) \equiv \inj(G,H')\pmod p$, 
as required.

Note that the largest graph~$G$ considered in the above inductive argument has the same number 
of vertices as~$H$.  So if $H$ and~$H'$ are not isomorphic then there must be a graph~$G$ 
with at most as many vertices as $H$ that distinguishes $H$ and~$H'$, that is,
$\hom(G,H)\not\equiv\hom(G,H')\pmod p$.
\end{proof}

\begin{proof}[Proof of Theorem~\ref{thmConfluence}]
Suppose $G\rightarrow_p^* G^*$ and $G\rightarrow_p^* G^\dag$, where $G^*$ and $G^\dag$
have no automorphisms of order~$p$.
Theorem~\ref{thm:inv} says the reduction operation 
$\rightarrow_p$ preserves the mod-$p$ Lov\'asz vector,
so $G^*$ and $G^\dag$ have the same vector.   
On the other hand, Lemma~\ref{lemModLovasz} 
above says that the mod-$p$ Lov\'asz vector characterises (isomorphism classes of) graphs 
with no automorphisms of order~$p$, so $G^*$ and $G^\dag$ are isomorphic.
\end{proof}

\section{Pinning colours to vertices} \label{secPinning}

We would like to be able to count the number of $H$-colourings of a given graph~$G$ 
in which certain vertices of~$G$ are forced to receive certain colours from~$H$.
This would allow us to isolate a suitable ``hard'' subgraph~$H'$ of~$H$, and hence
reduce the known hard $H'$-colouring problem to the particular 
$H$-colouring problem that interests us.
We achieve this by building gadgets, which are graphs with a distinguished
vertex, with the following property:  effectively, only a certain set of colours can be 
applied to the distinguished vertex of a gadget. 
By attaching these gadgets to a vertex  
of~$G$, we can restrict that vertex to be coloured with a particular set of colours.

\begin{defn}A {\em rooted graph} is a pair $(G,v)$ where $G$ is a graph 
and $v \in V(G)$ is a distinguished vertex of~$G$ (referred to as the {\em root}).
\end{defn}

In essence, we want to show that for any two distinct colours $h_1, h_2\in V(H)$ in a given~$H$, 
there exists some rooted graph $(\Gamma,\gamma)$ such that the number of ways of $H$-colouring $\Gamma$ 
with $\gamma$ receiving $h_1$ is different, \textit{modulo}~2, to the number of ways of $H$-colouring 
$\Gamma$ with $\gamma$ receiving~$h_2$. 
(In fact, as we can see, we can find such a rooted graph~$\Gamma$ for all prime moduli.) 
Suppose $G$ is an instance graph with distinguished root vertex~$v$.
We can then use rooted graphs such as $(\Gamma, \gamma)$ 
to pick out the colourings of~$G$ in which vertex~$v$ receives a colour from some particular 
subset of the colours.
Roughly, we do this by attaching a copy of $\Gamma$ to~$G$, 
identifying $\gamma$ and~$v$.  Call the resulting graph~$G'$.
Suppose a colouring of~$G$ with vertex~$v$ receiving $h_1$ extends to a colouring 
of $G'$ in (say) an odd number of ways.  Then a colouring with $v$ receiving $h_2$
will extend in an even number of ways.  In this way we have effectively ``cancelled'' the colourings
of~$G$ with $v$ coloured $h_2$, while leaving untouched those with $v$ coloured~$h_1$.

The construction of the required gadgets rests on a rooted version of Lemma~\ref{lemModLovasz}
Before we give the proof, we need to define rooted versions of a few concepts we have
already encountered. 

\begin{defn}A {\em homomorphism} (repectively, {\em isomorphism}) 
between two rooted graphs $(G,v)$ and $(G',v')$ is a graph 
homomorphism (respectively, isomorphism) 
$\phi:V(G)\rightarrow V(G')$ with $\phi(v) = v'$. 
An {\em automorphism\/} of rooted graph\/ $(G,v)$ is an isomorphism of $(G,v)$ 
to itself.
\end{defn}

\begin{defn}
We denote the number of homomorphisms from rooted graph $(G,g)$ to rooted graph 
$(H,h)$ by\/ $\homstar((G,g), (H,h))$.
If the roots are implied by the context we will sometimes suppress them in the above notation, 
and just write\/ $\homstar(G,H)$.
 
Similarly, we denote the number of injective homomorphisms from rooted graph $(G,g)$ to $(H,h)$ 
by\/ $\injstar((G,g),(H,h))$ and, again, we may suppress the specified vertices if they are implied by the context, 
instead writing\/ $\injstar(G,H)$.
\end{defn}

Finally, we will use the concept of the Lov\'asz vector of a rooted graph. 
For us, this will be the vector which counts, for a given rooted graph $(H,h)$, 
the number of homomorphisms to $(H,h)$ from every other rooted graph.

\begin{defn} \label{defnLovaszVectorRooted}
Let $G_1, G_2, \ldots$ be a fixed enumeration of all pairwise non-isomorphic rooted graphs.
Then the mod-$p$ Lov\'asz vector of a rooted graph~$H$ is the sequence $([\homstar(G_i,H)]_p:i \geq 1)$.
\end{defn}

We will use {\it parity Lov\'asz vector\/} is an alternative name for mod-2 Lov\'asz vector.

\begin{lem} \label{lemModLovaszRooted}
Suppose $p$ is a prime, and  $H$ and $H'$ are two rooted graphs 
neither of which have an automorphism of order $p$.
Then $H$ and $H'$ are isomorphic if and only if they have the same mod-$p$ Lov\'asz vector.
\end{lem}

\begin{proof}
As for Lemma~\ref{lemModLovasz}, but with $\homstar$ and $\injstar$ replacing $\hom$ and $\inj$.
In defining the quotient of a rooted graph $(G,g)$ by a partition~$\Theta=\{S_i:i\in I\}$, we define the 
root of $(G,g)/\Theta$ to be the vertex $i\in I$ such that $g\in S_i$.
\end{proof}

As with Lemma~\ref{lemModLovasz}, it can be seen that we need only finitely many terms of 
the mod-$p$ Lov\'asz vector to reconstruct $(H,h)$.

\subsection{Building Gadgets} \label{subsecGadgets}

In the following we return to \parityHcol, and are only interested in automorphisms of order two, or involutions. Note that many of the results in this section can be generalised to automorphisms of arbitrary prime order, but we only require the gadgets for the case $p=2$ in Section \ref{secTrees}, so only this case is presented here, for simplicity.

It will be useful to consider the case where $H$ and $H'$ have the same underlying graph but different roots (note that for $H$ and $H'$ to be non-isomorphic as rooted graphs, there can be no automorphism of~$H$ with takes $h$ to~$h'$, \ie, that $h$ and $h'$ lie in different orbits of the automorphism group of $H$). Since we will no longer be able to use the previous naming convention for the specified vertices, we will refer to the two roots 
in~$H$ as $x$ and~$y$. In the following, we will be assuming that $H$ is involution-free. As we saw in Section \ref{secReductionByInvolutions}, it suffices to consider the complexity of \parityHcol{} for involution-free~$H$.

Lemma \ref{lemModLovaszRooted} allows us to construct the following useful gadgets: 
given an involution-free graph~$H$ and two colours $x$ and $y$ which are in different orbits of $\Aut(H)$, there is a rooted graph $(\Gamma,\gamma)$ that 
distinguishes $x$ and~$y$.

\begin{lem} \label{lemHxHy}
 Given an involution-free graph $H$ and two vertices $x$ and $y$ which lie in different orbits of $\Aut(H)$, there exists a rooted graph $(\Gamma,\gamma)$ such that $\homstar((\Gamma,\gamma),(H,x)) \not \equiv \homstar((\Gamma,\gamma),(H,y)) \pmod 2$. 
\end{lem}

\begin{proof}
 Since $(H,x)$ and $(H,y)$ are non-isomorphic as rooted graphs, they have different parity Lov\'asz vectors 
 by Lemma~\ref{lemModLovaszRooted}. Simply take $(\Gamma,\gamma)$ to be the first rooted graph for which the corresponding entries of the parity Lov\'asz vectors of $(H,x)$ and $(H,y)$ differ. 
\end{proof}

We will use rooted graphs such as those guaranteed by Lemma~\ref{lemHxHy} 
as ``gadgets'' in a reduction from the problem of counting restricted $H$-colourings 
(in which a given vertex of the instance graph is forced to be coloured with colours from a specified orbit of $\Aut(H)$)
to the problem of counting unrestricted $H$-colourings \textit{modulo}~$2$.

\begin{thm} \label{thmPinning}
Given an involution-free graph~$H$, an orbit $O$ of the automorphism group of~$H$, 
and an oracle for \parityHcol, it is possible to determine, in polynomial time, 
the parity of the number of $H$-colourings of a rooted graph~$G$ in which the root 
receives a colour from~$O$.
\end{thm}

Note that this result would follow immediately if were able to build a gadget (\ie, rooted graph) $(\Gamma,\gamma)$ such that $\homstar((\Gamma,\gamma),(H,x))$ is odd, while $\homstar((\Gamma,\gamma),(H,y))$ is even for all $y \neq x$. Then we could just attach a copy of $\Gamma$ at the vertex of $G$ that we want to colour with $x$, identifying this vertex with $\gamma$, and then count $H$-colourings of the new graph. 
Unfortunately, Lemma \ref{lemHxHy} doesn't allow us to construct such a gadget,
as it doesn't allow us to choose which colour is~$x$ and which is~$y$.
However, we can construct a series of gadgets which allow us to count colourings of~$G$ in which the root of~$G$ 
receives a colour from a given orbit of~$H$, 
by developing a sort of algebra on the gadgets, as described below.

\begin{defn}Suppose $H$ is a graph, and $h_1,\ldots,h_n$ is an enumeration of the vertices of~$H$.
With each gadget\/ $(\Gamma,\gamma)$ we associate a vector $v_H(\Gamma)\in\mathrm{GF}(2)^n$,
indexed by $\{1,\ldots,n\}$, such that the $i^\text{th}$ component of the vector is $1$ 
if there are an odd number of $H$-colourings 
of\/ $\Gamma$ which use colour $h_i$ at $\gamma$, and $0$ otherwise.
\end{defn}

Note that if two colours (vertices of $H$) $h_i$ and $h_j$ are in the same orbit of the automorphism group 
of~$H$ then the $i^\text{th}$ and $j^\text{th}$ entries of $v_H(G)$ are the same for all rooted graphs~$G$. 
So we may instead consider the vector $v_{H}^*(G)$ which is indexed by {\it orbits\/} of the automorphism group of~$H$ 
rather than individual vertices of $H$, the coordinate of $v_H^*(G)$ associated with a given orbit 
being the coordinate of $v_H(G)$ associated with any (and hence all) of the colours in that orbit.
Note that $v_H(G)$ and $v^*_H(G)$ contain exactly the same information.

We define an operation that combines two rooted graphs by identifying their root vertices.
\begin{defn}
Given two rooted graphs $\Gamma$ and $\Pi$, we define the the new rooted graph  $\Gamma \cdot \Pi$ to be the graph obtained by identifying the roots of each. 
The root of\/ $\Gamma \cdot \Pi$ is the vertex formed by identifying the roots of the other two graphs.
\end{defn}
If we think of each gadget $\Gamma$ and $\Pi$ as enforcing a certain set of allowed colours at its root vertex, 
which can think of this operation as forming a gadget that enforces the intersection of these sets.
This is equivalent to saying that vector associated with the new gadget is 
obtained by taking the coordinate-wise product of the vectors associated with the individual gadgets.

\begin{defn} \label{defnStar}
 We define the operation $*:\mathrm{GF}(2)^n \times\mathrm{GF}(2)^n \rightarrow \mathrm{GF}(2)^n$ to be the coordinate-wise product of two vectors, 
 so the $i^\text{th}$ coordinate of $v*w$ is the $i^\text{th}$ coordinate of $v$ multiplied by the $i^\text{th}$ coordinate of $w$.
\end{defn}

\begin{lem} \label{lemGadgetStarGadget}
Suppose $\Gamma$ and $\Pi$ are two rooted graphs, and $H$ is graph. 
Then $v_H(\Gamma \cdot \Pi) =v_{H}(\Gamma)*v_{H}(\Pi)$.
\end{lem}

\begin{proof}
Fix a colour $h_i\in V(H)$.  The number of colourings of $\Gamma\cdot\Pi$ with the root receiving colour~$h_i$
is just the product of the number of colourings $\Gamma$ and $\Pi$ with the roots in each case receiving colour~$h_i$.
Thus, if  there is a zero in the $i^\text{th}$ place of either of the vectors $v_H^*(\Gamma)$ or $v_H^*(\Pi)$, 
 then there is a zero in the $i^\text{th}$ place of $v_H^*(\Gamma\cdot\Pi)$;
 otherwise there is a one. 
\end{proof}

We now introduce a formal sum of rooted graphs, with coefficients in $\GF(2)$, 
which preserves addition of these vectors. Note that since this sum has coefficients in $\GF(2)$ we have $\Gamma + \Gamma = 0$. 

\begin{defn} \label{defnPlusGraphs}
For a set of rooted graphs\/ $\Gamma_1, \Gamma_2, \cdots,\Gamma_r$, we define 
$\vH(\Gamma_1 + \Gamma_2 + \cdots + \Gamma_r)$ to be $\vH(\Gamma_1) + \vH(\Gamma_2) + \cdots + \vH(\Gamma_r)$.
\end{defn}

\begin{defn}
We will say that a vector $v\in \mathrm{GF}(2)^n$ is {\em implementable} for some $n$-vertex $H$ if there is a set of rooted graphs $\{\Gamma_1,\Gamma_2, \ldots, \Gamma_r\}$ such that $v$ is equal to $\vH(\Gamma_1 + \Gamma_2 + \cdots +\Gamma_r)$.
\end{defn}

\begin{lem} \label{lemImplementableVectorsClosed}
The set of vectors that are implementable for a given $H$ is closed under the operations of vector addition 
and point-wise multiplication (or the operation~$*$, as defined in Definition \ref{defnStar}).
\end{lem}

\begin{proof}
Suppose $v=\vH(\Gamma_1 + \Gamma_2 + \cdots + \Gamma_r)$ and $v'=\vH(\Pi_1 + \Pi_2 + \cdots + \Pi_s)$
are any two implementable vectors.  Then $v+v'$ is implementable, since
$$v+v'=\vH(\Gamma_1 + \Gamma_2 + \cdots + \Gamma_r+ \Pi_1 + \Pi_2 + \cdots + \Pi_s).$$
Furthermore,
\begin{align}
v*v'&= \vH(\Gamma_1 + \Gamma_2 + \cdots + \Gamma_r)*\vH(\Pi_1 + \Pi_2 + \cdots + \Pi_s)\notag\\
&= \big(\vH(\Gamma_1)+\vH(\Gamma_2)+\cdots+\vH(\Gamma_r)\big)*
    \big(\vH(\Pi_1) + \vH(\Pi_2) + \cdots + \vH(\Pi_s)\big)\notag\\
&= \vH(\Gamma_1)*\vH(\Pi_1) + \vH(\Gamma_1)*\vH(\Pi_2) + \cdots + \vH(\Gamma_r)*\vH(\Pi_s)\notag\\
&= \vH(\Gamma_1\cdot \Pi_1) + \vH(\Gamma_1 \cdot \Pi_2) + \cdots 
   + \vH(\Gamma_r\cdot \Pi_s)\label{eq:uselemGadgetStarGadget}\\
&= \vH(\Gamma_1 \cdot \Pi_1 + \cdots + \Gamma_r\cdot \Pi_s),\notag
\end{align}
where equality~(\ref{eq:uselemGadgetStarGadget}) follows from repeated application of Lemma~\ref{lemGadgetStarGadget}.
\end{proof}

\begin{lem} \label{lemCanImplementOnesGadgets}
 For any involution free graph, $H$, the all-ones vector is implementable, and for any pair of distinct orbits in $H$ there is at least one implementable vector which has a 1 at every vertex in one of the two orbits and a 0 at every vertex in the other orbit.
\end{lem}
\begin{proof}
 The all-ones vector is implementable using the graph on one vertex. The rooted graphs whose vectors distinguish between distinct orbits of colours in $H$ are obtained using Lemma \ref{lemHxHy}.
\end{proof}

We now know that the set of implementable vectors is closed under the operations of coordinate-wise addition and coordinate-wise multiplication, that for any pair of colours which are not in the same orbit it contains a vector which has different entries at these two places, and that it contains the all-ones vector. In the following lemma, we prove that these facts are enough to enable us to count colourings in which a specific vertex is required to be coloured with colours from any given orbit of Aut($H$). 

\begin{lem} \label{lemCanGetBasisVectors}
Consider a set, $S$, of vectors in $GF(2)^n$ which contains the all-ones vector $(1,1,\ldots,1)$ and has the property that for any two indices $i$ and $j$ there is some vector in the set whose $i^\text{th}$ coordinate differs from its $j^\text{th}$ coordinate. The closure of this set under the operations of coordinate-wise multiplication and coordinate-wise addition includes each of the vectors in the standard basis.
\end{lem}

\begin{proof}
We proceed by induction on~$n$. If $n=1$ the lemma clearly holds, as the all-ones vector is the only vector in the standard basis. Now, assume that $n>1$;  we shall attempt to construct the vectors in the standard basis in $\GF(2)^n$.

By induction, we can construct vectors that agree with the standard basis in the first $n-1$ places, without being able to control what happens in the $n^\text{th}$ place (note that the restriction of the set of vectors $S$ to the first $n-1$ places still satisfies the conditions of the lemma). That is, we can certainly obtain vectors of each of the following forms, where the $x_i$ can be either 0 or 1
\begin{displaymath}
\begin{array}{llllllll}
(1 & 1 & 1 & 1 & \ldots & 1 & 1 & 1) \\
(1 & 0 & 0 & 0 & \ldots & 0 & 0 & x_1) \\
(0 & 1 & 0 & 0 & \ldots & 0 & 0 & x_2) \\
(0 & 0 & 1 & 0 & \ldots & 0 & 0 & x_3) \\
\>\>\vdots&&&&&&&\,\vdots\\
(0 & 0 & 0 & 0 & \ldots & 0 & 1 & x_n) 
\end{array}
\end{displaymath}
This leaves several cases:

Case 1. The $x_i$ are all equal to zero. In this case, we already have the first $n-1$ vectors from the standard basis, and we can just take the sum of all $n-1$ vectors with the all-ones vector, which has a 1 in the last place and zeros everywhere else, to get the last one.

Case 2. There are at least two $i,j$ such that $x_i, x_j = 1$. But then the product of these two vectors is the vector $(0,0,\ldots,0,1)$. To obtain the remaining vectors from the standard basis, we just take the sum of this vector with any of those from the original list which had a 1 in the $n^\text{th}$ place, \ie, $e_i$ is the sum of this vector with the vector which had a 1 in the $i^\text{th}$ place and a 1 in the $n^\text{th}$ place.

Case 3. There is exactly one vector in the list, $v$ with a 1 as the $n^\text{th}$ coordinate. Say this vector has a 1 in the $i^\text{th}$ and $n^\text{th}$ places. By assumption, there is some vector in $S$ which has different values in the $n^\text{th}$ and $i^\text{th}$ places. The product of this with $v$ is a vector with exactly one 1, in either the $i^\text{th}$ or the $n^\text{th}$ place, and the sum of this basis vector with $v$ is the other of $e_i$ and $e_n$.
\end{proof}

\begin{lem} \label{lemCanImplementBasisVectors}
For any involution-free graph $H$, and any orbit of $O$ of $\Aut(H)$, the characteristic vector of~$O$
(which is 1 in coordinates indexed by~$O$ and 0 elsewhere) is implementable.
\end{lem}

\begin{proof}
for the purposes of this proof, it is convenient to think in terms to the abbreviated vectors $\vH^*(G)$ in place of the 
full vectors $\vH(G)$.  (This is not an essential change;  we are merely eliminating duplicated coordinates.)  So, now,
an implementable vector is one the form $\vH^*(\Gamma_1+\cdots+\Gamma_r)$, for some rooted graphs $\Gamma_1,\ldots, \Gamma_r$.
By Lemma \ref{lemImplementableVectorsClosed} the set of vectors we can implement is closed under the operations of addition and coordinate-wise multiplication,
and by Lemma \ref{lemCanImplementOnesGadgets} we can implement the all ones vector and, for each pair of indices (orbits) $i$ and $j$ a vector 
$v$ with $v_i\not=v_j$. Thus, by Lemma \ref{lemCanGetBasisVectors}, every vector in the standard basis is implementable.  
\end{proof}

We are now ready to return to Theorem \ref{thmPinning}.  Let $v$ be the characteristic vector of the 
orbit~$O$.  We know that $v$ is implementable.  So we now just have to 
show that our definition of ``implementable'' actually does what we want it to do. 
That is, it is possible to determine, 
in polynomial time using an oracle for unrestricted $H$-colourings, 
the parity of the number of $H$-colourings of a rooted graph~$G$ in which the root receives a colour from~$O$.

\begin{proof}[Proof of Theorem \ref{thmPinning}]
Let $v\in\mathrm{GF}(2)^n$ be the characteristic vector of the orbit $O$.
By Lemma \ref{lemCanImplementBasisVectors}, the vector $v$ is implementable, \ie, 
$v = v_H(\Gamma_1 + \Gamma_2 + \cdots + \Gamma_r)$ for some set of rooted graphs $\{\Gamma_1, \ldots, \Gamma_r\}$. 
Thus,
\begin{align*}
v_H(G)*v&=v_H(G)*v_H(\Gamma_1 + \Gamma_2 + \cdots + \Gamma_r)\\
&=v_H(G)*v_H(\Gamma_1)+\cdots+v_H(G)*v_H(\Gamma_r)\\
&=v_H(G\cdot \Gamma_1)+\cdots+v_H(G\cdot\Gamma_r).
\end{align*}
Now take the sum of the coordinates of the vectors, {\it modulo\/} 2:
$$
\sum_{i=1}^n(v_H(G)*v)_i=\sum_{i=1}^nv_H(G\cdot \Gamma_1)_i+\cdots+\sum_{i=1}^nv_H(G\cdot \Gamma_r)_i.
$$ 
The left-hand side counts, {\it modulo\/} 2, $H$-colourings of $G$ in which vertex $x$ receives a colour
from~$O$;  this is exactly the quantity we are interested in computing.
The $j^\text{th}$ term on the right hand side, counts, {\it modulo\/}~2, the number of (unrestricted)
$H$-colourings of the graph $G\cdot \Gamma_j$.  So the right-hand side can be evaluated 
using $r$~calls to an oracle for \parityHcol.
\end{proof}
 
Finally, we need an analogue of Theorem \ref{thmPinning} which allows pinning of two vertices of~$G$.
(We thank the authors of \cite{GobelEtAl13} for pointing out a lacuna at this point in an earlier 
version of the proof.)
 
\begin{cor} \label{corPinning}
Suppose $G$ is a graph with distinguished vertices $x$ and $y$.  
Given an involution-free graph~$H$, orbits $O$ and $O'$ of the automorphism group of~$H$, 
and an oracle for \parityHcol, it is possible to determine, in polynomial time, 
the parity of the number of $H$-colourings of~$G$ in which $x$ (respectively $y$)
receives a colour from~$O$ (respectively $O'$).
\end{cor}

\begin{proof}
Define the matrix $A=(a_{ij})\in \mathrm{GF}(2)^{n\times n}$ as follows. For all $1\leq i,j\leq n$,
$$
a_{ij}=\big[\text{number of colourings of $G$ with $x$ receiving colour $i$ and $y$ colour $j$}\big]_2.
$$
Let $u$ and $v$ be the  characteristic vectors of $O$ and $O'$.  By Lemma~\ref{lemCanImplementBasisVectors}
we know that $u$ and $v$ are implementable, \ie, 
$u=v_H(\Gamma_1)+\cdots+v_H(\Gamma_r)$ and $v=v_H(\Gamma'_1)+\cdots+v_H(\Gamma'_{s})$ for
some rooted graphs $\Gamma_1,\ldots, \Gamma_r$ and $\Gamma'_1,\ldots,\Gamma'_{s}$.  Thus
\begin{align*}
u^\intercal Av
&\equiv\big(v_H(\Gamma_1)+\cdots+v_H(\Gamma_r)\big)^\intercal A\big(v_H(\Gamma'_1)+\cdots+v_H(\Gamma'_{s})\big)\\
&\equiv \sum_{i,j=1}^n v_H(\Gamma_i)^\intercal A\,v_H(\Gamma'_j) \pmod2.
\end{align*}
Note that the left hand side is the quantity we are interested in, namely the number of restricted 
$H$-colourings of~$G$. Finally note that the $(i,j)^\mathrm{th}$ term in the last sum is equal, {\em modulo\/}~2, to the number 
of colourings of~$G$ with $\Gamma_i$ attached to~$x$ and $\Gamma'_j$ to~$y$.  So each term on the 
right hand side may be computed using an oracle for \parityHcol.
\end{proof}

\section{Trees} \label{secTrees}

As we have seen, if we apply the reduction operations defined in Definition \ref{defReduction} to any graph~$H$, this preserves the parity of the number of $H$-colourings of any graph~$G$. In particular, if a given $H$ reduces to a graph, say $H'$ such that the $H'$-colouring problem lies in~P, then the $H$-colouring problem also lies in~P\null. There are certain involution-free graphs~$H$ for which the $H$-colouring problem obviously lies in~P.

\begin{lem} \label{lemEasyTrees}
 Counting the number of $H$-colourings of a given graph $G$ can be done in polynomial time if $H$ is one of the null graph (the graph on no vertices), the graph on one vertex with no loop, the graph on one vertex with a loop, or the graph on two disconnected vertices, one with a loop and one without.
\end{lem}

\begin{proof}
 If $H$ is the null graph then there is no $H$-colouring of $G$, so the counting problem is obviously trivial. If $H$ is the graph on one vertex then $G$ has exactly one $H$-colouring if and only if $G$ has no edges, and zero otherwise, which can be determined in polynomial time. If $H$ is the graph on one vertex with a loop, then there is exactly one $H$-colouring of~$G$. If $H$ is the graph on two vertices one with a loop and one without then there are exactly $2^{|\mathrm{Isol}(G)|}$ colourings of~$G$, where $\mathrm{Isol}(G)$ is the set of isolated vertices of~$G$. Each isolated vertex can be coloured with either the looped vertex or the unlooped vertex of~$H$ independently, and all the vertices which form part of a connected component of size greater than one must be coloured with the looped vertex.
\end{proof}

\begin{cor} \label{corEasyGraphs}
If the reduced form associated with a given $H$ in the reduction system defined in Definition \ref{defReduction} is one of the null graph, the graph on one vertex, the graph on one vertex with a loop or the graph on two vertices, one with a loop and one without, then $H$-colouring is in~P.
\end{cor}

\begin{proof}
 This follows directly from Lemma \ref{lemEasyTrees} and the fact that the reduction system preserves the number of $H$-colourings, as shown in Lemma \ref{lemInvolutions}
\end{proof}

We conjecture that for general graphs, the reduction given in Corollary \ref{corEasyGraphs}, that is, $H$ reducing by involutions to one of the four trivial graphs is the only way in which the \parityHcol{} problem can fail to be $\parityP$-complete. Note that this criterion does encompass all of the easy cases identified by Dyer and Greenhill \cite{DyerGreenhill2000}. A complete graph with loops everywhere reduces to the null graph if it has an even number of vertices and the graph on one vertex with a loop if it has an odd number. On the other hand, a complete bipartite graph reduces to the graph on one vertex if there are an odd number of vertices in total, and the null graph otherwise.

In this section, we will prove that this conjecture is true for trees. In particular, if the reduced form, in the reduction system of Definition~\ref{defReduction}, associated with a given tree~$T$ is the graph on one vertex or the null graph, then the associated \parityTcol{} problem can be solved in polynomial time. Otherwise, it is $\parityP$-complete. Note that G\"obel, Goldberg and Richerby~\cite{GobelEtAl13} have recently extended the known range of validity of the conjecture from trees to cactus graphs. 

\subsection{Involution-Free Trees} \label{secInvFreeTrees}

Involution-free trees have quite a lot of structure, and we will exploit this when we build gadgets for our reductions from \parityIndSet{} (defined below) to \parityHcol{} in the next section.

\begin{lem} \label{lemVertsDeg2}
An involution-free tree on more than one vertex has two vertices of degree 2 which are adjacent to leaves.
\end{lem}

\begin{proof}
The argument given below is very similar to the standard argument given to show that any tree has at least two leaves.

The first observation to make is that any involution-free tree contains some path of length at least~3. If the maximum-length path in a tree is of length~1, then the tree consists of a single edge, and so has an involution. If it is of length~2, then the tree is a  star, and exchanging any two of its leaves is an involution.

Consider a longest path in an involution-free tree, and label the vertices of this path $p_0,p_1, \ldots, p_\ell$. Note that $p_0$ and $p_\ell$ are both leaves. Then we claim that both vertices $p_1$ and $p_{\ell-1}$ are degree 2. Note that $p_1$ and $p_{\ell-1}$ are in fact distinct vertices, as $\ell\geq 3$. Assume the degree of $p_1$ is greater than 2, and consider a vertex, $v$, adjacent to $p_1$ which is neither $p_0$ nor $p_2$. This vertex cannot have any neighbours which are not already in the path (as this would contradict maximality of the path). It also cannot have any neighbours which are in the path (as this would create a cycle, contradicting the fact that $G$ is a tree). Therefore, it cannot have any neighbours other than $p_1$. But  then exchanging this vertex with $p_0$ is an involution of $G$, so there is no such vertex, and $p_1$ is degree 2 as claimed. An analogous argument shows that $p_{\ell-1}$ must be degree 2.
\end{proof}

We will also require the following lemma.

\begin{lem} \label{lemInvFreeMeansAsymmetric}
An involution-free tree has trivial automorphism group.
\end{lem}
This follows directly from a characterisation of P\'olya \cite{Pol37} after Jordan \cite{Jor1869} of the automorphism groups of trees.

\begin{proof}
The automorphism group of a tree can be formed from symmetric groups using the operations of direct product and wreath product with a symmetric group~\cite{Pol37}.
Since the symmetric groups $S_n$ for $n>1$ have even order, the automorphism group of a tree is either of even order or has order 1. If it has even order, then by Lagrange's Theorem, the tree has an involution. So an involution-free tree has trivial automorphism group.
\end{proof}

Finally, we require the following technical lemma concerning the number of walks of various lengths between vertices in involution-free trees.

\begin{lem} \label{lemPathsOfLengthk}
Let $H$ be an involution-free tree, let $e_0$ be a vertex of degree 2 which is adjacent to a leaf in $H$, and let $e_\ell$ be a vertex of even degree such that there are no vertices of even degree on the path joining $e_0$ and $e_\ell$, where $\ell \geq 1$ is the length of the path joining $e_0$ and $e_\ell$. We will name the vertices on this path $e_0,o_1,o_2,\ldots, o_{\ell-1}, e_\ell$. 

Then there are an even number of vertices $v$ such that both:  
\begin{enumerate}
 \item $v$ is a neighbour of the first vertex on this path other than $e_0$, \ie,  
 $v$ is a neighbour of $e_1$ in the case $\ell=1$, and a neighbour of $o_1$ otherwise;  and 
 \item the number of walks of length $\ell$ from $v$ to $e_\ell$ in $H$ is odd.
\end{enumerate}
\end{lem}

\begin{proof}
We will refer in this proof to the vertices $o_1$ and $o_2$, which do not exist if $\ell=1$ or $\ell=2$, we deal with this at the end of this proof. For now, assume $\ell\geq 3$. We want to prove that there are an even number of neighbours of $o_1$ from which there are an odd number of walks of length $\ell$ to $e_\ell$ in $H$. There are an odd number of paths of length $\ell$ from $e_\ell$ to each of the neighbours of $o_1$ other than $o_2$: there is, in fact, one such walk, and it is the unique path connecting the neighbour to $e_\ell$ in the tree. We claim that there are an even number of walks of length $\ell$ from $e_\ell$ to $o_2$. 

A walk of length $\ell$ from $e_\ell$ to $o_2$ traverses exactly 1 edge more than once, as there is a unique path of length $\ell-2$ from $e_\ell$ to $o_2$. Two such walks which traverse the same edge more than once are identical. There is therefore a one-to-one correspondence between these walks and the edges which are traversed at least twice by at least one of them. We claim that the number of such edges is even.

Any edge which is adjacent to any of the vertices in $\{o_2, o_3 \ldots e_\ell\}$, and only those edges, may be traversed more than once, so it suffices to show that there are an even number of such edges. To see this, note that the only edges in this set which are adjacent to more than one of the vertices in the set are: $\{(o_2,o_3),(o_3,o_4), \ldots, (o_{\ell-1},e_\ell)\}$, there are the same number of edges in this set as the number of vertices of odd degree in $\{o_2, \ldots, e_\ell\}$. The total number of edges is then just the sum of the vertex degrees minus the number of edges which are adjacent to more than one of the vertices; but the sum of the vertex degrees is $\ell-2 \pmod 2$ (as there are $\ell-2$ vertices of odd degree) and the number of repeated edges is $\ell-2$, so the parity of the total number of edges is $(\ell-2) - (\ell-2) \equiv 0 \pmod 2$.

As noted above, if $\ell = 1$  or if $\ell=2$ the vertices $o_1$ or $o_2$ may not exist. However, the theorem still holds.

In particular, if $\ell=1$ then we actually have two adjacent vertices of even degree and the first vertex on the path which is not $e_0$ is in fact $e_1$, which is of even degree. Clearly there are an even number of vertices adjacent to $e_1$ with an odd number of length 1 walks to $e_1$, these being exactly the neighbours of $e_1$.

If $\ell = 2$, then again the vertex whose neighbours we are interested in is of odd degree, call it $o_1$, and there are an odd number of walks of length $2$ from $e_2$ to each of the neighbours of $o_1$ other than itself: in fact, there is exactly one such walk, the path joining the two vertices. On the other hand, $e_2$ is of even degree, so there are an even number of walks of length 2 from $e_2$ to itself. Since $o_1$ has an odd number of neighbours, this leaves an even number of neighbours of $o_1$ which have an odd number of length 2 walks to $e_2$, as claimed.
\end{proof}

\subsection{The Reduction} \label{secReduction}
Our starting point is the following problem, which was shown by Valiant~\cite{Valiant06} (in the guise of ``Mon 2-CNF'')
to be $\parityP$-complete; see also Faben~\cite[Thm.~3.5]{Faben2008}).

\prob%
{\parityIndSet.}%
{An undirected graph $G$.}%
{The parity of the number of independent sets in~$G$.}

\begin{thm} \label{thmHardness}
 Given an involution-free tree $H$ with more than one vertex, \parityHcol{} is $\parityP$-complete. In fact, there is a polynomial-time reduction from \parityIndSet{} to \parityHcol.
\end{thm}

\begin{defn}
 Given a graph $G$, we call $\sigma_2(G)$ the graph obtained by replacing every edge in $G$ with a path of length 2. We refer to the newly introduced vertices as {\em stretch vertices}, and the original vertices of $G$ as {\em $G$-vertices}. The construction is illustrated in Figure \ref{fig2stretch}.
\end{defn}

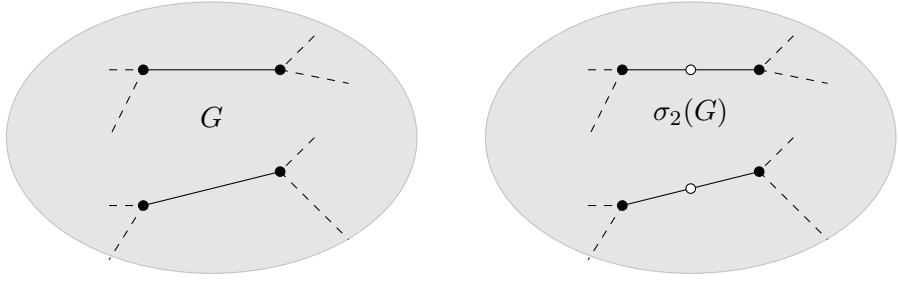
\begin{figure}
\begin{center}
\begin{tikzpicture}[scale = 0.9]

\draw [fill = gray, opacity=0.2] (3,3) ellipse (3 and 2);
\path (3,3) node[above] (Glab) {$G$};

\path (2,4) coordinate (a);
\draw [fill] (a) circle (2pt);

\path (1.5,4) coordinate (adash);
\path (1.5,3) coordinate (aprime);
\draw [dashed] (adash)--(a)--(aprime);

\path (4,4) coordinate (b);
\draw [fill] (b) circle (2pt);
\path (4.5,4.5) coordinate (bdash);
\path (5,3.8) coordinate (bprime);

\draw [dashed] (bdash)--(b)--(bprime);
\draw (a)--(b);

\path (2,2) coordinate (a1);
\draw [fill] (a1) circle (2pt);

\path (1.5,2) coordinate (adash1);
\path (1.5,1.2) coordinate (aprime1);
\draw [dashed] (adash1)--(a1)--(aprime1);

\path (4,2.5) coordinate (b1);
\draw [fill] (b1) circle (2pt);
\path (4.5,3) coordinate (bdash1);
\path (5,1.5) coordinate (bprime1);

\draw [dashed] (bdash1)--(b1)--(bprime1);
\draw (a1)--(b1);

\draw [fill = gray, opacity=0.2] (10,3) ellipse (3 and 2);
\path (10,3) node[above] (Glab') {$\sigma_2(G)$};

\path (9,4) coordinate (a');
\draw [fill] (a') circle (2pt);

\path (8.5,4) coordinate (adash');
\path (8.5,3) coordinate (aprime');
\draw [dashed] (adash')--(a')--(aprime');

\path (11,4) coordinate (b');
\draw [fill] (b') circle (2pt);
\path (11.5,4.5) coordinate (bdash');
\path (12,3.8) coordinate (bprime');

\draw [dashed] (bdash')--(b')--(bprime');
\draw (a')--(b');

\path (9,2) coordinate (a1');
\draw [fill] (a1') circle (2pt);

\path (8.5,2) coordinate (adash1');
\path (8.5,1.2) coordinate (aprime1');
\draw [dashed] (adash1')--(a1')--(aprime1');

\path (11,2.5) coordinate (b1');
\draw [fill] (b1') circle (2pt);
\path (11.5,3) coordinate (bdash1');
\path (12,1.5) coordinate (bprime1');

\draw [dashed] (bdash1')--(b1')--(bprime1');
\draw (a1')--(b1');

\path (10,4) coordinate (ab);
\draw [fill=white] (ab) circle (2pt);

\path (10,2.25) coordinate (a1b1');
\draw [fill=white] (a1b1') circle (2pt);
\end{tikzpicture}

\end{center}
\caption{The 2-stretch of $G$}
 \label{fig2stretch}
\end{figure}

The graph defined above, $\sigma_2(G)$, is usually referred to as the {\it 2-stretch\/} of $G$, and it is an established result that counting $H$-colourings of $\sigma_2(G)$ is equivalent to counting $H^2$-colourings of $G$, where $H^2$ is the multigraph whose adjacency matrix is the square of the adjacency matrix of $H$ (see, \eg,\cite{DyerGreenhill2000}).
We will use a variant of this stretch operation in which we count only those colourings of $\sigma_2(G)$ in which both the stretch vertices and the $G$-vertices are coloured with specific subsets of the colours in~$H$. This is achieved using gadgetry based on the principles established in Section~\ref{secPinning}.

We now detail the reduction from \parityIndSet.  First, given any graph~$G$, we will construct a certain graph $G^*$. We then claim that the number of $H$-colourings of $G^*$, with certain vertices restricted to receive certain colours from~$H$, is congruent \textit{modulo} 2 to the number of independent sets in~$G$. 

For a given involution-free tree $H$, pick a vertex of degree 2, $e_0$, adjacent to a leaf, and a vertex of even degree, $e_k$ such that the unique path of length $k$ in $H$ from $e_0$ to $e_k$ does not contain any vertex of even degree (exactly as in the statement of Lemma \ref{lemPathsOfLengthk}). Note that, as $H$ is involution-free, there are two vertices of even degree, and at least one vertex of degree two which is adjacent to a leaf in $H$ by Lemma \ref{lemVertsDeg2}, and we can choose $e_0$ and $e_k$ with the above properties. 

Now, given a graph $G$, first create $\sigma_2(G)$, then add two new vertices $R$ and $B$. Add an edge between each of the original vertices of $G$ ($G$-vertices) and $R$, and a path of length $k$ from every one of the new vertices (stretch vertices) of $\sigma_2(G)$ to $B$. We call this new graph $G^*$, and the construction is illustrated in Figure~\ref{figGStar}.

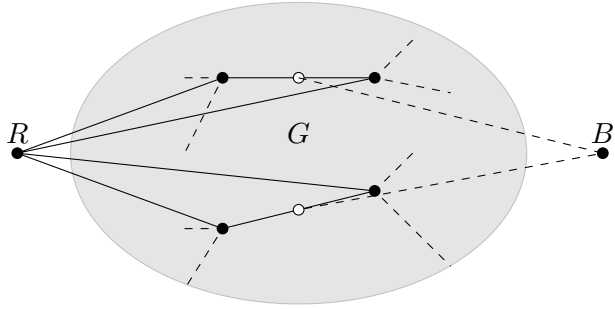
\begin{figure} 
 \begin{center}
  \begin{tikzpicture}

\path (0.3,3) coordinate (R);
\draw [fill] (R) circle (2pt);
\path (0.3,3) node[above] (Rlab) {$R$};

\path (8,3) coordinate (B);
\draw [fill] (B) circle (2pt);
\path (8,3) node[above] (Rlab) {$B$};

\draw [fill = gray, opacity=0.2] (4,3) ellipse (3 and 2);
\path (4,3) node[above] (Glab) {$G$};

\path (3,4) coordinate (a);
\draw [fill] (a) circle (2pt);

\path (2.5,4) coordinate (adash);
\path (2.5,3) coordinate (aprime);
\draw [dashed] (adash)--(a)--(aprime);

\path (5,4) coordinate (b);
\draw [fill] (b) circle (2pt);
\path (5.5,4.5) coordinate (bdash);
\path (6,3.8) coordinate (bprime);

\draw [dashed] (bdash)--(b)--(bprime);
\draw (a)--(b);

\path (3,2) coordinate (a1);
\draw [fill] (a1) circle (2pt);

\path (2.5,2) coordinate (adash1);
\path (2.5,1.2) coordinate (aprime1);
\draw [dashed] (adash1)--(a1)--(aprime1);

\path (5,2.5) coordinate (b1);
\draw [fill] (b1) circle (2pt);
\path (5.5,3) coordinate (bdash1);
\path (6,1.5) coordinate (bprime1);

\draw [dashed] (bdash1)--(b1)--(bprime1);
\draw (a1)--(b1);

\path (4,4) coordinate (ab);
\draw [fill=white] (ab) circle (2pt);

\path (4,2.25) coordinate (a1b1);
\draw [fill=white] (a1b1) circle (2pt);

\draw (a)--(R)--(b);
\draw (b1)--(R)--(a1);

\draw [dashed] (ab)--(B)--(a1b1);

  \end{tikzpicture}

 \end{center}
\caption{The construction of $G^*$}
\label{figGStar}
\end{figure}

Now, using the technology described in Corollary~\ref{corPinning}, and the fact that the orbit of a vertex in an involution-free tree is trivial by Lemma \ref{lemInvFreeMeansAsymmetric}, we can determine the parity of the number of $H$-colourings of $G^*$, in which $R$ is restricted to be coloured with~$e_0$ and $B$ is restricted to be coloured with~$e_k$, using only a \parityHcol{} oracle. We claim that this number is congruent (\textit{modulo} 2) to the number of independent sets in $G$. We will use what we know about the number of walks of length~$k$ between the colours $e_0$ and $e_k$ from Lemma \ref{lemPathsOfLengthk}.

\begin{lem} \label{lemReductionToIS}
Suppose $H$ is an involution-free tree, and let $e_0$ be a vertex of degree 2 adjacent to a leaf, and $e_k$ a vertex of even degree at distance $k \geq 1$ from~$e_0$ such that there are no vertices of even degree on the path of length~$k$ joining them.  Suppose $G$ is a graph and let $G^*$ be constructed from~$G$ as described above.
 
Then the number of $H$-colourings of $G^*$ in which $R$ receives~$e_0$ and $B$ receives~$e_k$ is congruent \textit{modulo}~2 to the number of independent sets in~$G$.
\end{lem}

\begin{proof}
First consider the $G$-vertices in $G$. They are all neighbours of a vertex which is coloured with $e_0$, so they must therefore receive colours that are adjacent to~$e_0$ in~$H$. But $e_0$ was chosen to be one of the vertices of degree 2 adjacent to a leaf in $H$, so $G$-vertices can only be coloured with either the leaf adjacent to $e_0$ (which we will call $l$) or with the first vertex on the path linking $e_0$ and $e_k$, which we will call $v_1$ in the remainder of this proof. This vertex is $o_1$, except in the case $k=1$ where it is $e_1$.

Now, consider the stretch vertices. These are connected to a vertex which is coloured $e_k$ by a path of length $k$. So, consider the colour used at a given stretch vertex, $s$. If there are an even number of walks of length $k$ from $e_k$ to this colour in $H$, then there are an even number of colourings of $G^*$ which use that colour at $s$, as there are an even number of ways of colouring the path joining $s$ and $B$, and the total number of colourings is the product of the number of ways of colouring this path with the number of ways of colouring the rest of the graph.

We therefore need to count colourings of~$G^*$ in which the colours used at the stretch vertices are such that there are an odd number of paths of length $k$ between them and $e_k$ in~$H$.
Note that these colours must also be adjacent to either $v_1$ or $l$ in~$H$ (as the $G$-vertices are all coloured with either $v_1$ or $l$, and every stretch vertex is adjacent to a $G$-vertex), and therefore, in fact, must be adjacent to~$v_1$, as the only neighbour of $l$ is $e_0$, which is also a neighbour of~$v_1$.

Now, we are reduced to considering colourings of $G^*$ in which the following conditions hold. The $G$-vertices are coloured either $l$ or $v_1$, while the stretch vertices are coloured with one of the neighbours of $v_1$ which has an odd number of length $k$ walks from itself to $e_k$. We claim that the parity of the number of such colourings is equal to the parity of the number of ways of colouring $G$ with the two colours $l$ and $v_1$ such that no two vertices coloured with $v_1$ are adjacent.

Consider a colouring of $G$ with the colours $v_1$ and $l$. If there are two vertices of $G$ which are adjacent in $G$ and both coloured with $v_1$ then there are an even number of extensions of this colouring to an $H$-colouring of $G^*$: the stretch vertex between the two $G$-vertices in $G^*$ can be coloured with any one of the neighbours of $v_1$ which are at distance $k$ from $e_k$ in $H$, and there are an even number of such vertices by Lemma \ref{lemPathsOfLengthk}.

On the other hand, if there are no two such vertices, there is exactly one extension of the given colouring of $G$ to an $H$-colouring of $G^*$: every one of the stretch vertices is adjacent to a vertex which is coloured $l$, so the stretch vertices must all be coloured $e_0$, and as there is only one path of length $k$ from $e_0$ to $e_k$ in $H$, this determines the colouring of the vertices on the paths linking the stretch vertices to $B$.

So the number of colourings of $G^*$ with $H$ such that $R$ is coloured $e_0$ and $B$ is coloured $e_k$ is congruent \textit{modulo} 2 to the number of colourings of $G$ in which each vertex is either coloured with $l$ or $v_1$ and adjacent vertices may not both be coloured with $v_1$. But these are exactly the independent sets of $G$: vertices coloured $v_1$ are ``in'' the independent set and vertices coloured $l$ are ``out''.
\end{proof}

\begin{proof} [Proof of Theorem \ref{thmHardness}]
 By Theorem \ref{corPinning} and Lemma \ref{lemInvFreeMeansAsymmetric} we can count $H$-colourings of $G^*$ in which $R$ is coloured $e_0$ and $B$ is coloured $e_k$ in polynomial time if equipped with an $H$-colouring oracle. But we know that the number of such colourings is congruent \textit{modulo} 2 to the number of independent sets in $G$. Since clearly $G^*$ can be constructed from $G$ in polynomial time, this gives us a polynomial-time Turing reduction from \parityIndSet{} to \parityHcol.
\end{proof}

\subsection{A Dichotomy for Trees} \label{subsecTreeDichotomy}
The main result now follows easily.

\begin{proof}[Proof of Theorem~\ref{thmTreeDichotomy}]
By Lemma \ref{lemInvolutions}, the number of $H$-colourings of a graph $G$ is congruent \textit{modulo} 2 to the number of $H'$-colourings, where $H'$ is any graph obtained from $H$ by reducing $H$ by any of its involutions. Also, if $H$ is a tree then any graph $H'$ which can be reached from~$H$ by reduction by involutions is also a tree. It therefore suffices to consider involution-free trees.

If $H$ is an involution-free tree, and $H$ contains more than one vertex, then Theorem~\ref{thmHardness} shows that \parityHcol{} is $\parityP$-complete. On the other hand, if $H$ contains either 0 or 1 vertices then \Hcol{} (and hence \parityHcol) is polynomial-time solvable by Lemma~\ref{lemEasyTrees}.
\end{proof}

Note that although the dichotomy is certainly decidable, it is not clear whether it can be decided in polynomial time.
On the face of it, finding the reduced form associated with a graph~$H$ requires finding an involution of~$H$, 
and no polynomial-time algorithm is known for this problem.

\section{Other graphs} \label{secOtherGraphs}
As noted earlier, we conjecture not only that there is a dichotomy for the complexity of \parityHcol{} for general $H$, but that this dichotomy is the same as that for trees. In other words, that the only way in which a \parityHcol{} problem can be polynomial-time solvable is if $H$ reduces by involutions to one of the four trivial graphs.
We now show that we can restrict our attention to connected $H$. That is, if an involution-free graph $H$ has any connected component $H_1$ for which \parityHoneCol{} is $\parityP$-hard, then the parity colouring problem associated with~$H$ is itself $\parityP$-hard.

\begin{thm} \label{thmConnectedComponentHardness}
Let $H$ be an involution-free graph. If $H_1$ is a connected component of $H$ and \parityHoneCol{} is $\parityP$-hard, then \parityHcol{} is $\parityP$-hard.
\end{thm}

\begin{proof}
 Take any graph $G$, and assume that $G$ is connected (since the number of $H$-colourings of $G$ is just the product of the number of $H$-colourings of each of its connected components). We can use an oracle for $H$-colouring to determine the parity of the number of colourings of~$G$ in which only colours from~$H_1$ are used in the following way: let $v \in V(G)$ be any vertex of $G$. For each colour $h_i \in V(H_1)$, we can count the colourings of $G$ in which $v$ is coloured~$h_i$ using Theorem \ref{thmPinning}. Notice that the size of the orbit of~$h_i$ in $\Aut(H)$ is odd, as $H$ has no involutions, so the parity of the number of colourings of $G$ with $h_i$ at $v$ is the same as the parity of the number of colourings of $G$ which use any of the vertices in the orbit of~$h_i$ at~$v$. 

But we can do this for every vertex in $H_1$, and since $G$ is connected, any colouring which uses a vertex from $H_1$ at $v$ can use only colours from $H_1$ anywhere in $G$. Conversely, any colouring of $G$ which uses only colours from $H_1$ must use some colour from $H_1$ at $v$, so this does indeed allow us to count all such colourings of $G$.
\end{proof}

Note that this actually allows us to strengthen Theorem \ref{thmTreeDichotomy}: the $H$-colouring problem associated with any \textit{forest} $H$ is polynomial-time solvable if the reduced form associated with the forest in the reduction system described in Section \ref{secReductionByInvolutions} is the null graph or the graph on one vertex, and $\parityP$-complete otherwise.

\bibliographystyle{plain}
\bibliography{CSPbib}

\end{document}